\documentclass[a4paper]{article}

\usepackage[headinclude]{typearea}
\usepackage[english]{babel}           
\usepackage[utf8]{inputenc}
\usepackage[T1]{fontenc}
\usepackage{scrtime}
\usepackage{amsmath,amsthm,amsfonts,amssymb}
\usepackage {color}                                  
\usepackage {pifont}                                 
\usepackage {multicol}                               
\usepackage {times}                                  
\usepackage[numbers]{natbib}                                                     
\usepackage [scaled=.90]{helvet}                     
\usepackage {courier}                                
\usepackage {graphicx}                               
\usepackage {array}                                  
\usepackage {longtable}                              
\usepackage {makeidx}                                
\makeindex
\usepackage{fancyvrb}
\usepackage{enumerate} 
\usepackage{ifpdf}
\usepackage{caption}

\usepackage{setspace}

\usepackage{marvosym}

\theoremstyle{plain}
\newtheorem{theorem}{Theorem}[section]
\newtheorem{cor}[theorem]{Corollary}
\newtheorem{lemma}[theorem]{Lemma}
\newtheorem{proposition}[theorem]{Proposition}

\theoremstyle{definition}

\newtheorem{definition}[theorem]{Definition}
\newtheorem{remark}[theorem]{Remark}

\newcommand{\E}{{\mathbb{E}}}

\newcommand{\N}{{\mathbb{N}}}
\renewcommand{\P}{{\mathbb{P}}}

\newcommand{\R}{{\mathbb{R}}}

\newcommand{\ub}{{\underline b}}
\newcommand{\uy}{{\underline y}}
\newcommand{\uz}{{\underline z}}
\newcommand{\unu}{{\underline \nu}}

\renewcommand{\P}{{\mathbb P}}

\newcommand{\cD}{{\cal D}}
\newcommand{\cE}{{\cal E}}
\newcommand{\cF}{{\cal F}}

\newcommand{\cL}{{\cal L}}
\newcommand{\cN}{{\cal N}}

\newcommand{\cS}{{\cal S}}
\newcommand{\cT}{{\cal T}}

\newcommand{\be}{\begin{equation}}
\newcommand{\ee}{\end{equation}}
\newcommand{\bea}{\begin{eqnarray}}
\newcommand{\eea}{\end{eqnarray}}
\newcommand{\beast}{\begin{eqnarray*}}
\newcommand{\eeast}{\end{eqnarray*}}
\newcommand{\bproof}{\begin{proof}}
\newcommand{\eproof}{\end{proof}}

\newcommand{\fpi}{p}
\newcommand{\fPi}{\Pi}

\title{Dividend maximization in a hidden Markov switching model}

\author{Michaela Sz\"olgyenyi}

\begin{document}

\date{Preprint, February 2016}

\maketitle


\begin{abstract}
In this paper we study the valuation problem of an insurance company by maximizing the expected discounted future dividend payments in a model with partial information that allows for a changing economic environment.
The surplus process is modeled as a Brownian motion with drift. This drift depends on an underlying Markov chain the current state of which is assumed to be unobservable. The different states of the Markov chain thereby represent different phases of the economy.\\
We apply results from filtering theory to overcome uncertainty and then we give an analytic characterization of the optimal value function. Finally, we present a numerical study covering various scenarios to get a clear picture of how dividends should be paid out.
\end{abstract}

\vspace{5mm}


\centerline{\underline{\hspace*{16cm}}}

M. Sz\"olgyenyi\\
Institute of Statistics and Mathematics, Vienna University of Economics and Business, 1020 Vienna, Austria\\
michaela.szoelgyenyi@wu.ac.at

\centerline{\underline{\hspace*{16cm}}}


\section{Introduction}

The classical risk measure for insurance companies is the probability of ruin.
This quantity is intensively studied, since in 1903 \citet{lundberg1903} introduced his model for the surplus of an insurance company.
However, in many models a trajectory can only either lead to ruin or tend to infinity.
Hence, \citet{definetti1957} introduced expected discounted future dividend payments as a valuation principle for a homogeneous insurance portfolio,
which builds an alternative risk measure.
This concept originally comes from corporate finance, where the firm value is often determined by the accumulated future dividend payments.

Our aim is to solve the valuation problem of an insurance company in these terms. Since its introduction, the dividend maximization problem has been solved in various setups, for example by \citet{shreve1984}, \citet{jeanblanc1995}, \citet{radner1996}, and \citet{asmussen1997} in diffusion models. \citet{albrecher20112} study a similar problem with random intervention times.
Overviews can be found in \citet{albrecher20092}, or \citet{avanzi2009}. For an introduction to optimization problems in insurance in general we refer to \citet{schmidli2008} and \citet{azcue2014}.\\

However, all of these contributions assume a constant economic environment. The dividend maximization problem is considered over a potentially long time horizon, which makes the assumption that the economy does not change a strong one. Since a changing economic environment structurally influences the insurance market, we would like to incorporate such changes into our model.
In setups which allow for a change in the economic environment the dividend maximization problem has been studied for example by \citet{jiang2012}, \citet{sotomayor2011}, \citet{zhu2013}, \citet{albrecher2012}, and \citet{azcue2015}.  \citet{jiang2012}, \citet{sotomayor2011}, and \citet{zhu2013} consider a diffusion model for the surplus, the parameters of which are driven by an underlying Markov chain. Such a setup is called {\em regime switching model}. In a Cram\'er-Lundberg type model \citet{albrecher2012} allow the economy to become worse once before ruin, and \citet{azcue2015} allow for a finite number of shifts of the economy. In all these models, the driving Markov chain is assumed to be {\em observable}. This means that they assume full information and therefore their models differ from the model studied herein.\\

We use the dividend maximization approach to solve the valuation problem of an insurance company in a {\em hidden} Markov model.
More precisely, we model the surplus process of the insurance company as a Brownian motion with drift. The drift is assumed to be driven by an underlying Markov chain, the current state of which is {\em unobservable} under the available information.
In contrast to \citet{sz12} and \citet{sz13}, where the setup was Bayesian, i.e., the underlying Markov chain was not allowed to change its state, we allow for regime shifts. This gives the model a different interpretation.
Here, the different states of the Markov chain represent different phases of the economy. It is of certain interest on the one hand to allow for changes in the economic environment and on the other hand it is realistic that the current state of the economy is not exactly known instantaneously, but only over time, and also over time the drift of a diffusion cannot be estimated satisfactorily, see \cite[Chapter 4.2]{rogers2013} -- even more if shifts are allowed.
Thus allowing for uncertainty extends \cite{sotomayor2011} to a more practically relevant direction.\\

In the literature of mathematical finance hidden Markov models have already been used intensively for studying investment problems, e.g., in \citet{karatzas2001}, \citet{sass2004}, \citet{rieder2005}, or \citet{frey2012, frey2014}. Also dividend problems were solved in the mathematical finance literature, see \citet{hubalek2004} and \citet{grandits2007}, who seek to maximize the expected accumulated utility of dividends, and the expected utility of accumulated dividends, respectively. However, in the insurance context related results concerning hidden Markov models are scarce. We refer to \citet{gerber1977} as an example, who models the value of a single insurance policy as a Brownian motion with unobservable drift describing uncertainty about the quality of a risk. Another, more recent example is the paper by \citet{liang2014}, who study optimal reinsurance and investment under unobservable claim size and intensity.
In \citet{decamps2015} a valuation problem similar to the dividend maximization problem is studied in a rather specific model from the point of view of corporate finance.\\

This paper is organized as follows.
In Section \ref{sec:Setup} we define our model and show how to overcome uncertainty by applying a result from stochastic filtering theory, and thus transform the setup into one under full information.
The stochastic optimization problem under study is presented in Section \ref{sec:Opti}.
For solving the stochastic optimization problem we derive the Hamilton-Jacobi-Bellman (HJB) equation and characterize the optimal value function as a viscosity solution to this HJB equation.
We also prove uniqueness of the viscosity solution even though there is a lack of boundary conditions.
In Section \ref{sec:Num} we treat the problem numerically.
First of all we examine the filter dynamics. Then, we solve the HJB equation numerically. For this we need to introduce a correction term to ensure positivity of the scheme, but we can show that the approximate solution converges.
We present a multitude of numerical examples.
Furthermore, we are able prove admissibility of the candidate optimal strategies, which means showing that the underlying system of stochastic differential equations with discontinuous drift and degenerate diffusion has a solution.
Section \ref{sec:Concl} concludes the paper.\\

The main contribution of this paper is the analytical characterization of the solution to the dividend maximization problem in a hidden Markov switching model, including a non-standard uniqueness proof of the generalized solution to the associated Hamilton-Jacobi-Bellman (HJB) equation, and an extensive numerical study of the outcomes of this model. The intention behind this study is to impart good comprehension of the model and of how to optimally pay out dividends. Assuming only partial information makes the model more natural and realistic.


\section{Setup and filtering}
\label{sec:Setup}

All stochastic variables introduced in the following are assumed to be defined on the filtered probability space $(\cE, \cF, \{\cF_t\}_{t \ge 0}, \P)$.\\
The surplus process is given by
\begin{equation}
\label{eq:dynX1}X_t=x+\int_0^t \mu_s \,ds+\sigma B_t - L_t\,,
\end{equation}
with initial capital $x>0$, where $\mu=(\mu_t)_{t \ge 0}$ is the unobservable drift process, $\sigma$
is the constant and known volatility, and $B=(B_t)_{t \ge 0}$ is a standard Brownian motion.
The accumulated dividend process $L$ is given by
\begin{align}
\label{eq:dynL1}dL_t=u_t \,dt\,,
\end{align}
with $L_0=0$ and density $u_t \in [0,K]$ for all $t\ge0$ that serves as the control variable in our optimization problem.
Note that the surplus process $X$ is always associated to a dividend strategy, however, for notational reasons, we will not make that explicit.\\
Furthermore, let $\mu_t:=\mu(Y_t)$, where $Y=(Y_t)_{t\ge 0}$ is an $M$-state Markov chain with known generator matrix $Q=(q_{ij})_{i,j=1}^M$. Let $\mu_t \in \lbrace \mu_1, \dots, \mu_M \rbrace$, where $\mu_t=\mu_i$, if $Y_t=e_i$, and $e_i$ is the $M$-dimensional unit vector the $i$-th component of which is $1$. Without loss of generality let $\mu_1 > \dots > \mu_M$. We assume that the current state of the Markov chain is unobservable under the observation filtration, but we know its initial distribution $\P(Y_0=e_i)=p_i$ with $p_i > 0$ for all $i \in \lbrace 1, \dots ,M\rbrace$ and $\sum_{i=1}^M p_i = 1$.\\

Note that it is crucial that the volatility is assumed not to be driven by the Markov chain, as then it would be possible to estimate the current state from the quadratic variation.\\

The uncontrolled surplus process $Z=(Z_t)_{t\geq 0}$ is given by
\begin{align}
\label{eq:dynZ1}Z_t&= x + \int_0^t \mu_s\, ds+\sigma B_t\,.
\end{align}
As the dividend payments have to be adapted to the uncontrolled process, the observation filtration is given by $\{\cF^{Z}_t\}_{t\ge 0} \subset \{\cF_t\}_{t\ge 0}$, which is the augmentation of the filtration generated by $Z$.\\

\subsection*{Stochastic filtering}
\label{subsec:filtering}

As $\mu_t$ is not $\cF^Z_t$-measurable, we are in a situation of partial information. To overcome this uncertainty, we apply a result from stochastic filtering theory. This means we replace the unobservable parameter $\mu_t$ by an estimator, which potentially uses all the information generated by $\cF^Z_t$, but not more. We refer the interested reader to \citet{elliott1995} for more information about hidden Markov models and their filtering, and to \citet{bain2009} for stochastic filtering in general. \citet{rieder2005} suggest using the Wonham filter (see \cite{liptser1977,wonham1964}) for the case where the unobservable variable is driven by a Markov chain.\\

From \citet[Theorem 9.1]{liptser1977} we know the following proposition.
\begin{proposition}
 Denote the conditional probability that the Markov chain is in state $i$ at time $t$ (and hence $\mu_t=\mu_i$) as
 \begin{align*}
  \pi_i(t)=\P(\mu_t=\mu_i \mid \cF_t^{Z})
 \end{align*}
 for $i= 1, \ldots, M$, and the estimator for the drift as
 \begin{align}\label{eq:estimator}
  \nu_t=\E(\mu_t \mid \cF_t^{Z})=\sum_{i=1}^M \mu_i \pi_i (t)\,.
 \end{align}
 Then $(\pi_1, \dots \pi_M)$ solves the following system of stochastic differential equations
 \begin{align}\label{eq:wonham}
  \pi_i(t)&=p_i+\int_0^t \sum_{j=1}^M q_{ji} \pi_j (s) \, ds + \int_0^t \pi_i (s) \frac{\mu_i-\nu_s}{\sigma} \, dW_s\,,\\
  \pi_i(0)&=p_i\,,
 \end{align}
 for $i=1,\ldots,M$, with the innovation process
 \begin{align}\label{eq:W}
  W_t=\int_0^t\frac{\mu_s-\nu_s}{\sigma} \,ds + B_t\,.
 \end{align}

 Furthermore, $W=(W_t)_{t \ge 0}$ is an $\{\cF_t^{Z}\}_{t \ge 0}$ - Brownian motion.
\end{proposition}

In particular, we have that the estimator $\nu=(\nu_t)_{t \ge 0}$ is adapted to the observation filtration.\\
Please note that $\pi_M(t)=1-\sum_{i=1}^{M-1} \pi_i(t)$ for all $t \ge 0$, which implies that the correct state space of the filter is the simplex $\bar \cS := \{(\pi_1,\dots,\pi_M) \in [0,1]^M: \sum_{i=1}^M \pi_i = 1\}$, the interior of which is denoted by $\cS := \{(\pi_1,\dots,\pi_M) \in (0,1)^M: \sum_{i=1}^M \pi_i = 1\}$.
For later use, we define $\fPi:=(\fPi_t)_{t \ge 0}=(\pi_1(t), \dots, \pi_{M-1}(t))_{t \ge 0}$ and $\fpi:=\fPi_0=(p_1, \dots, p_{M-1})$.

From now on we consider the following $M$-dimensional system of SDEs:
\begin{align}
\label{eq:dynX2}X_t&= x + \int_0^t (\nu_s-u_s) \,ds +\sigma W_t\,,\\
\label{eq:dynpi}\pi_i(t)&=p_i+\int_0^t \left(q_{Mi} + \sum_{j=1}^{M-1} (q_{ji}-q_{Mi}) \pi_j (s)\right) \, ds + \int_0^t \pi_i (s) \frac{\mu_i-\nu_s}{\sigma} \, dW_s \qquad i= 1, \dots, M-1\,,
\end{align}

where
\begin{align}\label{eq:nu}
\nu_t=\mu_M + \sum_{j=1}^{M-1}(\mu_j-\mu_M)\pi_j(t)\,.
\end{align}

Since there is only one source of uncertainty, which is adapted to the observation filtration, in system \eqref{eq:dynX2}, \eqref{eq:dynpi}, we are now in a situation of full information, but at the cost of $M-1$ additional dimensions.


\section{Stochastic optimization}
\label{sec:Opti}

In this section we  at first define the stochastic optimization problem under study. Then we derive the associated HJB equation.
Finally, we present the main result of this section, which is the characterization of the solution of the optimization problem as the unique viscosity solution to the HJB equation.\\

We would like to find the optimal value function $V$, which is the supremum over all dividend policies of the discounted dividend payments up to the time of ruin $\tau:=\inf\{t\ge 0 \, \vline \, X_t\le 0\}$,

\begin{align*}
V(x,\fpi)=\sup_{u\in A} J^{(u)}(x,\fpi)=\sup_{u\in A} \E_{x,\fpi} \left(\int_0^{\tau}e^{-\delta t}u_t \,dt\right)\,,
\end{align*}
where $\delta>0$ is the discount rate, $A$ denotes the set of admissible controls, and $\E_{x,\fpi}(\cdot)$ denotes the expectation under the conditions $X_0=x$ and $\fPi_0=\fpi$. Admissible controls are $\{\cF^Z_t\}_{t\ge 0}$-progressively measurable, $[0,K]$-valued, and fulfill $u_t \equiv 0$ for $t>\tau$.\\

Note that the underlying system of stochastic processes \eqref{eq:dynX2}, \eqref{eq:dynpi} describes autonomous state dynamics in the sense of \cite[Section IV.5]{fleming2006}. Furthermore, we will consider an infinite time horizon. Therefore, the optimal control will be Markovian.

\begin{lemma}\label{lem:x-infinity}
The optimal value function $V$ is continuous. We have $0\le V \le \frac{K}{\delta}$, $V$ increases in $x$, and $\lim_{x \to \infty} V(x,\fpi)=\frac{K}{\delta}$ uniformly in $\fpi$.
\end{lemma}

\begin{proof}
From \cite[Chapter 3, Theorem 5]{krylov1980} we know that the optimal value function $V$ is continuous.\\
The monotonicity of $V$ with respect to $x$ follows by an argument from \cite[Chapter 2.5.1, p. 97]{schmidli2008}.\\
Clearly, the optimal value function is bounded by $0\le V(x,\fpi)\le \int_0^\infty K e^{-\delta s}ds=\frac{K}{\delta}$, and it is easy to check that in the limit it converges to $\frac{K}{\delta}$, cf. \cite[Chapter 2.5.1, p. 97]{schmidli2008}.
\end{proof}

\subsection*{The Hamilton-Jacobi-Bellman equation}
\label{subsec:HJB}

For deriving the HJB equation we need a version of the dynamic programming principle, or Bellman principle, see \cite[Chapter 3, Theorem 6]{krylov1980}.
\begin{proposition}[Bellman principle]
\label{prop:bellman}
For every bounded stopping time $\eta$ we have
\[
V(x,\fpi)
=\sup_{u \in A} \E_{x,\fpi}\left(\int_0^{\tau \wedge \eta}e^{-\delta t}u_t\,dt+e^{-\delta(\tau\wedge\eta)}V(X_{\tau\wedge\eta},\fPi_{\tau\wedge\eta})\right)\,.
\]
\end{proposition}

Now, assuming $V \in C^2$, one can derive the associated HJB equation from the Bellman principle:
\begin{align}
 \label{eq:HJB}(\cL - \delta) V +\sup_{u \in [0,K]}(u(1-V_x))=0\,,
\end{align}
where
\begin{align*}
 \cL V&=  \mu_M V_x+\sum_{i=1}^{M-1} \left((\mu_i-\mu_M) \fpi_i\ V_x + \left(q_{Mi} + \sum_{j=1}^{M-1} (q_{ji}-q_{Mi}) \fpi_j  \right) V_{\fpi_i} + \fpi_i \left( \mu_i - \nu \right) V_{x \fpi_i}\right.\\
 &+ \left.\frac{1}{2}\sum_{k=1}^{M-1} \left( \left( \fpi_i \frac{\mu_i-\nu}{\sigma} \right) \left( \fpi_k \frac{\mu_k-\nu}{\sigma} \right) V_{\fpi_i \fpi_k} \right)\right)+ \frac{1}{2} \sigma^2 V_{xx}\,,
\end{align*}
and $\nu$ is given by \eqref{eq:nu}.
The HJB-equation is a second order degenerate-elliptic PDE, since there is only one Brownian motion driving the $M$-dimensional process $(X_t,\fPi_t)_{t \ge 0}$. The supremum in \eqref{eq:HJB} is attained at
\begin{align*}
 u=\begin{cases}
    K, & V_x \le 1\\
    0, & V_x > 1\,.
   \end{cases}
\end{align*}

As boundary conditions we have for $x=0$ and $x\rightarrow \infty$
\begin{align}
\label{eq:bcx1} V(0,\fpi)&=0\,,\\
\label{eq:bcx2} V(x,\fpi)&\rightarrow\frac{K}{\delta} \mbox{ uniformly in }\fpi\mbox{ as }x \rightarrow \infty\,.
\end{align}
For $\fpi_i \in \{0,1\}$, $i=1,\dots,M-1$, we have no boundary conditions available. However, as the filter $\fpi$ never reaches the boundary, boundary conditions in these particular directions are not required for the solution being well-defined.
The reason for this is that we still get uniqueness in the interior and on the relevant part of the boundary, see Corollary \ref{cor:comparison}.\\
It should be mentioned at this place that the solution to the problem with complete information, i.e., for observable $Y$, does not serve as boundary condition.
This is because even if we knew that we started in a certain state, then a moment later we would again not be able to observe the state, whereas in the model suggested by \cite{sotomayor2011},
one is then still able to observe it.

\subsection*{Analytic characterization}
\label{subsec:analytics}

Now we come to the analytic characterization of the optimal value function.
In the Markov switching setup where the current state of the Markov chain is observable (see \cite{sotomayor2011}), the HJB equation can be solved explicitly and the solution is smooth.
In our case the HJB equation \eqref{eq:HJB} is a degenerate-elliptic PDE, which makes the existence of a smooth solution questionable. Thus, we deal with a weaker concept of solutions, namely viscosity solutions. The only required smoothness for this is continuity, however, the concept is still strong enough to prove uniqueness. Furthermore, it is also useful for numerical treatment, see \citet{barles1991}, or \citet[Chapter IX]{fleming2006}. These are two important strengths of the concept of viscosity solutions and make it very beneficial to problems like ours.
Therefore, we are going to characterize the optimal value function $V$ as the unique viscosity solution of \eqref{eq:HJB}.\\

Denote by $\cT:= \{(\fpi_1,\dots,\fpi_{M-1}) \in (0,1)^{M-1}: \sum_{i=1}^{M-1} \fpi_i < 1\}$ the state space of the first $M-1$ dimensions of the filter with closure $\bar \cT:= \{(\fpi_1,\dots,\fpi_{M-1}) \in [0,1]^{M-1}: \sum_{i=1}^{M-1} \fpi_i \le 1\}$ and boundary $\partial \bar \cT$. Further denote
$\Omega:=(0,\infty) \times \cT$, $\bar{\Omega}=[0,\infty)\times \bar \cT$ and let $\partial \bar{\Omega}$ be its boundary. Furthermore, let $\Gamma_-:=(0,\infty)\times\partial \bar \cT \subseteq \partial \bar{\Omega}$.
Then $\Gamma_+:=\partial \bar \Omega \backslash \Gamma_-$ denotes the so-called {\it relevant part} of the boundary.

\begin{definition}\label{def:viscosity}
(viscosity solution)
\begin{enumerate}
 \item A function $w:\bar{\Omega} \rightarrow \R$ is a {\it viscosity subsolution} to \eqref{eq:HJB}, if
 \[
  (\cL-\delta) \phi(\bar{x},\bar{\fpi}) + \sup_{u \in [0,K]}(u(1-\phi_x(\bar{x},\bar{\fpi})))\ge 0
 \]
 for all $(\bar{x},\bar{\fpi}) \in \Omega$ and for all $\phi \in C^2(\Omega)$ such that $w-\phi$ attains a maximum at $(\bar{x},\bar{\fpi})$ with $w(\bar{x},\bar{\fpi})=\phi(\bar{x},\bar{\fpi})$.
 
 \item A function $w:\bar{\Omega} \rightarrow \R$ is a {\it viscosity supersolution} to \eqref{eq:HJB}, if
 \[
  (\cL-\delta) \psi(\bar{x},\bar{\fpi}) + \sup_{u \in [0,K]}(u(1-\psi_x(\bar{x},\bar{\fpi})))\le 0
 \]
 for all $(\bar{x},\bar{\fpi}) \in \Omega$ and for all $\psi \in C^2(\Omega)$ such that $w-\psi$ attains a minimum at $(\bar{x},\bar{\fpi})$ with $w(\bar{x},\bar{\fpi})=\psi(\bar{x},\bar{\fpi})$.
 
\item $w:\bar{\Omega} \rightarrow \R$ is a {\it viscosity solution} to \eqref{eq:HJB},  if it is both a viscosity sub- and supersolution.
\end{enumerate}
\end{definition}
The basic idea of viscosity solutions is to estimate the function from below and from above by smooth test functions.
For details about viscosity solutions, see, e.g., \citet{fleming2006}, or \citet{crandall1992}.\\

The following theorem shows the connection between the solution of the optimization problem and the weak solution of the HJB equation.
\begin{theorem}\label{thm:viscosity}
The optimal value function $V$ is a viscosity solution of \eqref{eq:HJB} with boundary conditions \eqref{eq:bcx1} and \eqref{eq:bcx2}.
\end{theorem}

\begin{proof}
 We have to show that the optimal value function $V$ is a viscosity sub- and supersolution, cf.~\cite[Proof of Theorem 5.1]{sz12}. \\

\underline{Viscosity supersolution:} Let $\psi \in C^2(\Omega)$, $\psi \le V$ and $(\bar{x},\bar{\fpi})$ such that $V(\bar{x},\bar{\fpi})=\psi(\bar{x},\bar{\fpi})$. Let $\eta >0$.\\
  Applying the dynamic programming principle we get
  \begin{align*}
   \psi(\bar{x},\bar{\fpi})=V(\bar{x},\bar{\fpi}) &= \sup_{u \in A} \E_{\bar{x},\bar{\fpi}} \left( \int_0^{\tau \wedge \eta} e^{-\delta t} u_t \,dt + e^{-\delta(\tau \wedge \eta)}V(X_{\tau \wedge \eta}, \fPi_{\tau \wedge \eta})\right)\\
   &\ge \E_{\bar{x},\bar{\fpi}} \left( u \frac{1-e^{-\delta(\tau \wedge \eta)}}{\delta} + e^{-\delta(\tau \wedge \eta)} \psi (X_{\tau \wedge \eta}, \fPi_{\tau \wedge \eta})\right)
  \end{align*}
  for any fixed $u \in [0,K]$.\\
  Now we apply It\^o's formula to $\psi$, note that the stochastic integrals are martingales, divide by $\eta$ and let $\eta \rightarrow 0$. This yields
  \[
   0 \ge u-\delta \psi (\bar{x},\bar{\fpi}) + \cL \psi (\bar{x},\bar{\fpi}) -u \, \psi_x (\bar{x},\bar{\fpi})\,.
  \]
  As $u$ was arbitrary,
  \[
   0 \ge (\cL-\delta) \psi (\bar{x},\bar{\fpi}) + \sup_{u \in [0,K]} \left(u (1-\psi_x (\bar{x},\bar{\fpi}))\right)\,.
  \]
  Thus, $V$ is a viscosity supersolution.\\
  
\underline{Viscosity subsolution:} Let $\phi \in C^2 (\Omega)$, $\phi \geq V$ and $(\bar{x}, \bar{\fpi})$ such that $\phi (\bar{x}, \bar{\fpi}) = V(\bar{x},\bar{\fpi})$.
  For $\varepsilon > 0$ let $\eta > 0$ and $u^{\varepsilon,\eta}$ be an $\frac{\varepsilon \eta}{2}$-optimal dividend
policy for the first part of the dynamic programming principle, and denote the surplus coming from $u^{\varepsilon,\eta}$ as $X^{\varepsilon,\eta}$. Then 

\begin{align*}
\phi(\bar{x}, \bar{\fpi}) - \frac{\varepsilon \eta}{2} &= V(\bar{x}, \bar{\fpi}) - \frac{\varepsilon \eta}{2} \le \E_{\bar{x},\bar{\fpi}} \left( \int_{0}^{\tau \wedge \eta} e^{- \delta t} u_t^{\varepsilon,\eta} \, d t + e^{-\delta(\tau \wedge \eta)} V(X_{\tau \wedge \eta}^{\varepsilon,\eta}, \fPi_{\tau \wedge \eta})\right)\\
& \le \E_{\bar{x},\bar{\fpi}} \left( \int_0^{\tau \wedge \eta} e^{- \delta t} u_t^{\varepsilon,\eta} dt + e^{-\delta(\tau \wedge \eta)} \phi (X_{\tau \wedge \eta}^{\varepsilon,\eta}, \fPi_{\tau \wedge \eta})\right)\\
&= \E_{\bar{x},\bar{\fpi}} \left( \int_{0}^{\tau \wedge \eta} e^{-\delta t} u_t^{\varepsilon,\eta} dt + e^{- \delta (\tau \wedge \eta)} \left(\phi (\bar{x}, \bar{\fpi})+ \int_0^{\tau \wedge \eta} \cL \phi \, dt - \int_0^{\tau \wedge \eta} \phi_x u_t^{\varepsilon,\eta} \, dt \right)\right)\\
&\le \E_{\bar{x},\bar{\fpi}} \left( \int_{0}^{\tau \wedge \eta} e^{-\delta t} \hat u_t^{\varepsilon,\eta}\, dt + e^{- \delta (\tau \wedge \eta)} \left(\phi (\bar{x}, \bar{\fpi})+ \int_0^{\tau \wedge \eta} \cL \phi \, dt - \int_0^{\tau \wedge \eta} \phi_x \hat u_t^{\varepsilon,\eta} \, dt \right)\right)+\frac{\varepsilon \eta}{2}\,,
\end{align*}
where $\hat u^{\varepsilon,\eta}$ is continuous in $t$, has values in $[0,K]$, and approximates $u^{\varepsilon, \eta}$ in $L^1([0,\eta),\lambda)$.
Furthermore, we applied It\^o's formula and used that the stochastic integrals are martingales.
Rearranging the inequality and dividing by $\eta$ yields
\begin{align*}
-\varepsilon  &\le \E_{\bar{x},\bar{\fpi}} \left( \frac{e^{- \delta (\tau \wedge \eta)}-1}{\eta} \phi (\bar{x}, \bar{\fpi})+\frac{ e^{- \delta (\tau \wedge \eta)}}{\eta} \int_0^{\tau \wedge \eta} \cL \phi \, dt
+ \frac{1}{\eta}\int_0^{\tau \wedge \eta} \left(e^{-\delta t}-e^{- \delta (\tau \wedge \eta)}\phi_x\right) \hat u_t^{\varepsilon,\eta} \, dt \right)\,.
\end{align*}
Now we apply the mean value theorem:
\begin{align*}
-\varepsilon  &\le \E_{\bar{x},\bar{\fpi}} \left( \frac{e^{- \delta (\tau \wedge \eta)}-1}{\eta} \phi (\bar{x}, \bar{\fpi})+\frac{e^{- \delta (\tau \wedge \eta)}}{\eta}  \int_0^{\tau \wedge \eta} \cL \phi \, dt
+ \frac{\tau \wedge \eta}{\eta}\left(e^{-\delta \xi}-e^{- \delta (\tau \wedge \eta)}\phi_x\right) \hat u_\xi^{\varepsilon,\eta} \right)\,,
\end{align*}
and let $\eta \rightarrow 0$ along a sequence:
\begin{align*}
-\varepsilon  &\le (\cL-\delta) \phi (\bar{x}, \bar{\fpi})+ \limsup_{\eta \to 0} \E_{\bar{x},\bar{\fpi}}\left( \frac{\tau \wedge \eta}{\eta}\left(e^{-\delta \xi}-e^{- \delta (\tau \wedge \eta)}\phi_x\right) \hat u_\xi^{\varepsilon,\eta} \right)\,.
\end{align*}
Fatou's lemma gives
\begin{align*}
& \limsup_{\eta \to 0}  \E_{\bar{x},\bar{\fpi}}\left( \frac{\tau \wedge \eta}{\eta}\left(e^{-\delta \xi}-e^{- \delta (\tau \wedge \eta)}\phi_x\right) \hat u_\xi^{\varepsilon,\eta} \right)
\le 
 \E_{\bar{x},\bar{\fpi}}\left( \limsup_{\eta \to 0} \frac{\tau \wedge \eta}{\eta}\left(e^{-\delta \xi}-e^{- \delta (\tau \wedge \eta)}\phi_x\right) \hat u_\xi^{\varepsilon,\eta} \right)\\
 &= \E_{\bar{x},\bar{\fpi}}\left( (1-\phi_x(\bar{x},\bar{\fpi})) \left(  \limsup_{\eta \to 0} \hat u_\xi^{\varepsilon,\eta} \,1_{\{1-\phi_x(\bar{x},\bar{\fpi}) \ge 0\}} +  \liminf_{\eta \to 0} \hat u_\xi^{\varepsilon,\eta} \,1_{\{1-\phi_x(\bar{x},\bar{\fpi}) < 0\}} \right)\right)
 = \tilde u(\bar x, \bar \fpi)(1-\phi_x(\bar{x},\bar{\fpi}))\,,
\end{align*}
where $\tilde u(\bar x, \bar \fpi)= \left(  \limsup_{\eta \to 0} \hat u_\xi^{\varepsilon,\eta} \,1_{\{1-\phi_x(\bar{x},\bar{\fpi}) \ge 0\}} +  \liminf_{\eta \to 0} \hat u_\xi^{\varepsilon,\eta} \,1_{\{1-\phi_x(\bar{x},\bar{\fpi}) < 0\}}\right) $.
As $\varepsilon>0$ was arbitrary,
\[
(\cL - \delta) \phi(\bar{x}, \bar{\fpi}) + \tilde u(\bar x, \bar \fpi)(1-\phi_x(\bar{x}, \bar{\fpi})) \geq 0\,.
\]
Since $\tilde u(\bar x, \bar \fpi)(1-\phi_x(\bar{x}, \bar{\fpi})) \leq \sup_{u \in [0, K]} u (1- \phi_x(\bar{x}, \bar{\fpi}))$, we get 
\[
(\cL - \delta) \phi(\bar{x}, \bar{\fpi}) + \sup_{u \in [0, K]} u (1- \phi_x(\bar{x}, \bar{\fpi})) \geq 0\,.
\]
Thus, V is also a viscosity subsolution.\\

Altogether, $V$ is a viscosity solution.
\end{proof}

Now it remains to prove uniqueness. For this, one has to prove a weak maximum principle, which in standard proofs results in the statement that if two viscosity solutions are equal on the boundary, they are also equal in the interior of the domain. However, as mentioned above, we have no boundary conditions available in the $\fpi$ directions. But \citet{lions1983b} shows that if the underlying stochastic process does not reach some parts of the boundary of the domain with a positive probability, then as these parts are not reached anyway, the study can be restricted to the interior and the relevant part of the boundary.

\begin{theorem}[Comparison]\label{thm:comparison}
Let $w_1$ and $w_2$ be bounded and continuous viscosity solutions of \eqref{eq:HJB}.\\ 
If $w_1\leq w_2$ on $\Gamma_+$ and $\lim_{x\to\infty} (w_1-w_2)(x,\fpi)\le 0$ uniformly in $\fpi$,
then $w_1\leq w_2$ on $\Omega$.
\end{theorem}

\begin{proof}
 Define $\tau':=\inf\{t \ge 0 \vert (X_t,\fPi_t) \in \partial \bar \Omega\}$. We need to check whether
 $\P(\tau'<\infty \,,\,\, (X_{\tau'},\fPi_{\tau'}) \in \Gamma_-)=0$.
 From \cite[Corollary 2.2]{chigansky2007} we know that the Wonham filter never reaches the boundary. Therefore, the above probability is indeed zero.
Hence, we may apply \cite[Corollary II.1]{lions1983b}, which proves the statement.
\end{proof}

Uniqueness of the solution of \eqref{eq:HJB} now follows as a corollary.
\begin{cor}\label{cor:comparison}
The optimal value function $V$ is the unique bounded viscosity solution of \eqref{eq:HJB} on $\Omega \cup \Gamma_+$ with boundary conditions \eqref{eq:bcx1} and \eqref{eq:bcx2}.
\end{cor}

The following theorem shows that our analytic characterization includes smooth solutions to the HJB equation. Furthermore, it can be concluded that if there is a dividend policy leading to a smooth value function that solves the HJB equation in the viscosity sense, then this policy is optimal.

\begin{theorem}\label{thm:verification}
Let $w$ be a viscosity supersolution of \eqref{eq:HJB} with
boundary conditions \eqref{eq:bcx1} and \eqref{eq:bcx2}, and $w \in C^2$ almost everywhere. Then $V \le w$.
\end{theorem}
 
 \begin{proof}
 The proof runs along the same lines as \cite[Proof of Theorem 5.3]{sz12}.
We begin with convoluting $w$ with a Gauss Weierstrass kernel function. Due to notational ambiguities, we remark that here the area of a circle with radius $1$ is denoted by $\pi$.
Let $\varphi(x,\fpi) := \frac{1}{\pi^\frac{M}{2}} e^{-\left( x^2 + \sum_{i=1}^{M-1} \fpi_i^2 \right) }$ and
\[
 \varphi^n(x,\fpi) := n^M \int_{-\infty}^\infty \int_{- \infty}^\infty \dots \int_{- \infty}^\infty w(x-s,(\fpi_1-t_1,\dots \fpi_{M-1}-t_{M-1})) \varphi(ns,nt)\, ds\, dt_1\, \dots \, dt_{M-1}
\]
for $n \in \N$, where $nt=(nt_1,\dots,nt_{M-1})$.
Clearly, as $n \to \infty$,
$\varphi^n \to w$ and $\cL \varphi^n \to \cL w$, see \citet{wheeden1977}.\\

For an admissible strategy $u=(u_t)_{t \ge 0}$ and $T>0$,
 \begin{align*}
e^{-\delta (T \wedge \tau)} \varphi^n (X_{T \wedge \tau},\fPi_{T \wedge \tau})
=& \varphi^n(x,\fpi) + \int_0^{T \wedge \tau} e^{-\delta t}\, d\varphi^n(X_t,\fPi_t) + \int_0^{T \wedge \tau} \varphi^n(X_t,\fPi_t)\, d(e^{-\delta t})\\
=& \varphi^n(x,\fpi) + \int_0^{T \wedge \tau}  e^{-\delta t} \left[ -\delta \varphi^n(X_t,\fPi_t) + \cL \varphi^n (X_t, \fPi_t) - u_t \varphi_x^n (X_t, \fPi_t) \right]\,dt
+ M_t\,,
\end{align*}

where $M=(M_t)_{t\ge0}$ is a martingale. Therefore,
\begin{align*}
 \E_{x,\fpi} \left( e^{-\delta (T \wedge \tau)} \varphi^n (X_{T \wedge \tau},\fPi_{T \wedge \tau})\right)
= \varphi^n(x,\fpi) + \E_{x,\fpi} \left( \int_0^{T \wedge \tau}  e^{-\delta t} \left[ -\delta \varphi^n(X_t,\fPi_t) + \cL \varphi^n (X_t, \fPi_t) - u_t \varphi_x^n (X_t, \fPi_t) \right]\,dt\right)\,.
\end{align*}

Note that $w$ fulfills
\[
 -\delta w + \cL w + (1-w_x) u \le 0 \quad \mbox{a.e.}
\]
Thus, for $\varepsilon>0$ we can choose  $n$ large enough such that 
\[
 -\delta \varphi^n + \cL \varphi^n + (1-\varphi^n_x) u \le \varepsilon\,,
\]
and hence
\[
 \cL \varphi^n \le \delta \varphi^n - (1-\varphi^n_x) u + \varepsilon\,.
\]
Therefore,
\begin{align*}
&\E_{x,\fpi} \left( e^{-\delta (T \wedge \tau)} \varphi^n (X_{T \wedge \tau},\fPi_{T \wedge \tau})\right)\\
&\le \varphi^n(x,\fpi) + \E_{x,\fpi} \left( \int_0^{T \wedge \tau}  e^{-\delta t} \left[ -\delta \varphi^n(X_t,\fPi_t) + \delta \varphi^n(X_t,\fPi_t) - (1-\varphi_x^n(X_t,\fPi_t)) u_t + \varepsilon - u_t \varphi_x^n (X_t, \fPi_t) \right]\,dt\right)\\
&= \varphi^n(x,\fpi) - \E_{x,\fpi} \left( \int_0^{T \wedge \tau}  e^{-\delta t} u_t \,dt - \varepsilon \int_0^{T \wedge \tau}  e^{-\delta t} \,dt\right)\,.
\end{align*}

By dominated convergence we have for $n\rightarrow\infty$
\begin{align*}
\E_{x,\fpi} \left( e^{-\delta (T \wedge \tau)} w (X_{T \wedge \tau},\fPi_{T \wedge \tau})\right)
\le w(x,\fpi) - \E_{x,\fpi} \left( \int_0^{T \wedge \tau}  e^{-\delta t} u_t \,dt - \varepsilon \int_0^{T \wedge \tau}  e^{-\delta t} \,dt\right)\,.
\end{align*}
As $\varepsilon$ was arbitrary,
\[
\E_{x,\fpi} \left( e^{-\delta (T \wedge \tau)} w (X_{T \wedge \tau},\fPi_{T \wedge \tau})\right) \le w(x,\fpi) - \E_{x,\fpi} \left( \int_0^{T \wedge \tau}  e^{-\delta t} u_t \,dt\right)\,,
\]
and hence
\[
\E_{x,\fpi} \left( e^{-\delta (T \wedge \tau)} w (X_{T \wedge \tau},\fPi_{T \wedge \tau})\right) + \E_{x,\fpi} \left( \int_0^{T \wedge \tau}  e^{-\delta t} u_t \,dt\right) \le w(x,\fpi)\,.
\]

Since $w$ is bounded we have that $\lim_{T \to \infty} \E_{x,\fpi} \left(
e^{-\delta (T \wedge \tau)} w (X_{T \wedge \tau},\fPi_{T \wedge \tau})\right)
= 0$. Thus, by bounded convergence, 
\[
  J^{(u)}(x,\fpi)=\E_{x,\fpi} \left( \int_0^{\tau}  e^{-\delta t} u_t \,dt\right) 
= \lim_{T \to \infty} \E_{x,\fpi} \left( \int_0^{\tau\wedge T}  e^{-\delta t} u_t \,dt\right) 
\le w(x,\fpi)\,.
\]
Since for each $u$, $w$ dominates the optimal value function, by taking the supremum over all $u \in [0,K]$ in the derivation, we get
$V(x,\fpi) \le w(x,\fpi)$.
\end{proof}

\begin{remark}
If there is a strategy $\tilde{u} \in A$ such that $J^{(\tilde{u})}$ is a viscosity supersolution with $J^{(\tilde{u})} \in C^2$ almost everywhere,
then by Theorem \ref{thm:verification}, $J^{(\tilde{u})}=V$ is the classical solution to the problem, and $\tilde{u}$ is the optimal policy.
\end{remark}


\section{Numerics}
\label{sec:Num}

In this section we first simulate a path of the Markov chain $Y$ and a path of the Wonham filter to get an idea of its behaviour. Then we describe a numerical procedure for computing approximations to $V$ and the optimal dividend policy. We will restrict our numerical analysis to the case $M=2$. For a better understanding of the numerical results we transform the state process and consider $(X_t,\nu_t)_{t \ge 0}$, where $\nu_t=\mu_1 \pi_1(t) + \mu_2 (1-\pi_1(t))$, and the corresponding transformed HJB equation, instead of considering $(X_t,\pi_1(t))_{t \ge 0}$.\\

For simulating paths of the Wonham filter we need to express $W$ in terms of $Z$. The representation follows from \eqref{eq:W}:
$dW_t=\frac{dZ_t-\nu_t dt}{\sigma}$. 
With this and equations \eqref{eq:estimator} and \eqref{eq:wonham} we get
\begin{align}\label{eq:nu_sim}
 d \nu_t  &= \left( q_{11}(\nu_t-\mu_2) + q_{21} (\mu_1-\nu_t) - \frac{(\nu_t-\mu_2)(\mu_1-\nu_t)\nu_t}{\sigma^2} \right) dt + \frac{(\nu_t-\mu_2)(\mu_1-\nu_t)}{\sigma^2} dZ_t\,,\\
 \nu_0 &=:\upsilon=\mu_2 + p_1 (\mu_1-\mu_2)\,.
\end{align}

Now we simulate the increments of the Brownian motion $B$ as $\sqrt{\Delta t}\,\ub_t$, where $\ub_t \sim \cN(0,1)$. Furthermore, we simulate a path $\uy$ of the Markov chain $Y$, and calculate $d\uz_{t+\Delta t} = \mu (\uy_t) \Delta t + \sigma \sqrt{\Delta t}\,\ub_t$. With this and equation \eqref{eq:nu_sim} we are ready to calculate the path of the estimator
\begin{align*}
 d \unu_{t+\Delta t}  = \left( q_{11}(\unu_t-\mu_2) + q_{21} (\mu_1-\unu_t) - \frac{(\unu_t-\mu_2)(\mu_1-\unu_t)\unu_t}{\sigma^2} \right) \Delta t + \frac{(\unu_t-\mu_2)(\mu_1-\unu_t)}{\sigma^2} d\uz_{t+\Delta t}
\end{align*}
by applying the Euler-Maruyama scheme.

Figure \ref{fig:wonhamplot} shows a path of the drift of the uncontrolled process governed by the underlying Markov chain, and its estimator. We see that the estimator always needs some time to notice the change in the drift and only adapts to it slowly, but this clearly depends on the choice of $Q$. Furthermore, we see that the estimator indeed does not reach the boundary.\\
\begin{figure}[ht]
\begin{center}
\includegraphics[scale=0.6]{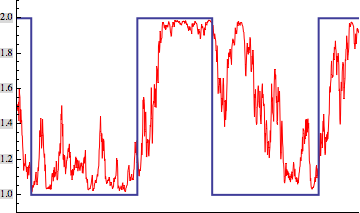}
\end{center}
\caption{Wonham filter estimate (red) of the drift (blue).}\label{fig:wonhamplot}
\end{figure}

Now we are going to solve the HJB equation numerically. The transformed HJB equation is
\begin{align}
 \label{eq:HJBnu}(\tilde \cL - \delta) V +\sup_{u \in [0,K]}(u(1-V_x))=0\,,
\end{align}
where
\begin{align*}
 \tilde \cL V&=  \upsilon V_x + \left(q_{21}(\mu_1-\upsilon)+q_{11}(\upsilon-\mu_2)\right) V_{\upsilon} + (\mu_1-\upsilon)(\upsilon-\mu_2) V_{x \upsilon}
 + \frac{1}{2 \sigma^2} (\mu_1-\upsilon)^2(\upsilon-\mu_2)^2 V_{\upsilon \upsilon} + \frac{1}{2} \sigma^2 V_{xx}\,.
\end{align*}

To solve this PDE numerically, we first introduce an approximation based on an idea from \cite{frey2014} to the HJB equation to ensure convergence of the scheme.
Define
\begin{equation}\label{eq:apprunderl}
\begin{aligned}
  X^{(u),\epsilon}_t&= x+\int_0^t (\nu_s^{\epsilon}-u_s) \, ds + \sigma W_t\,,\\
  \nu_t^{\epsilon} &=\upsilon +\int_0^t \left( q_{11}(\nu_s^{\epsilon}-\mu_2) + q_{21} (\mu_1-\nu_s^{\epsilon}) \right) \,ds + \int_0^t\frac{(\nu_s^{\epsilon}-\mu_2)(\mu_1-\nu_s^{\epsilon})}{\sigma}\, dW_s+2\int_0^t \sqrt{\bar \epsilon(\nu_s^{\epsilon})} \,d\widetilde W_s\,,
 \end{aligned}
 \end{equation}
 and $\tau^{\epsilon}=\inf\{t\ge0 | X_t^{(u),\epsilon} \le 0\}$. $\widetilde W=(\widetilde W_t)_{t \ge0}$ is a Brownian motion independent of $W$, and $\bar \epsilon$ is a smooth function that is bounded, has bounded derivatives, and is vanishing at $\mu_1,\mu_2$ and $\bar \epsilon(\upsilon)=\epsilon$ on $(\mu_1+\zeta,\mu_2-\zeta)$ for some small $\zeta>0$. Furthermore, denote by $J^{(u),\epsilon}, V^\epsilon$ the value function and the optimal value function associated to the approximate underlying system \eqref{eq:apprunderl}.
 This introduces an additional second order term in the approximate HJB equation
 \begin{align}
 \label{eq:HJBnuappr}(\tilde \cL+\bar \epsilon V_{\upsilon \upsilon}^\epsilon - \delta) V^\epsilon +\sup_{u \in [0,K]}(u(1-V^\epsilon_x))=0\,.
\end{align}
The difference to the approximation in \cite{frey2014} is that here we do not regularize our HJB equation with the additional term, but just use it to ensure positivity of the scheme in the interior of the computation domain. Note however that the analytic characterization of the optimal value function herein does not require regularization.\\
The following theorem states that the solution still converges to the optimal value function.

\begin{theorem}
 Let $V^\epsilon$ be the solution to \eqref{eq:HJBnuappr}. Then $\lim_{\epsilon \to 0}\|V^\epsilon-V\|=0$.
\end{theorem}

\begin{proof}
 We show that the result holds for every value function $J^{(u),\epsilon} \to J^{(u)}$. From this we may conclude that the result also holds for $V$, see \cite[Corollary 7.4]{frey2014}.\\

 Let $J^{(u),\epsilon,T}, J^{(u),T}$ denote the corresponding value functions stopped at time $0<T<\infty$. We have
 \begin{align*}
  \lim_{\epsilon\to 0} \|J^{(u),\epsilon} (x,\upsilon)-J^{(u)} (x,\upsilon)\| & = \lim_{\epsilon\to 0} \lim_{T \to \infty}\|J^{(u),\epsilon} -J^{(u)} \|\\
  &= \lim_{\epsilon\to 0} \lim_{T \to \infty}\|J^{(u),\epsilon}-J^{(u),\epsilon,T}+J^{(u),\epsilon,T}-J^{(u),T}+J^{(u),T} -J^{(u)} \|\\
  &\le \lim_{\epsilon\to 0} \lim_{T \to \infty}\left(\|J^{(u),\epsilon}-J^{(u),\epsilon,T}\|+\|J^{(u),\epsilon,T}-J^{(u),T}\|+\|J^{(u),T} -J^{(u)}\|\right)\\
  &= \lim_{\epsilon\to 0} \lim_{T \to \infty}\|J^{(u),\epsilon}-J^{(u),\epsilon,T}\|+\lim_{\epsilon\to 0} \lim_{T \to \infty}\|J^{(u),\epsilon,T}-J^{(u),T}\|\\ &+\lim_{\epsilon\to 0} \lim_{T \to \infty}\|J^{(u),T} -J^{(u)}\|\,,\\
 \end{align*}
 where we skipped the arguments in the calculations and used that all terms in the last but one row are bounded since the optimal value function is bounded due to Lemma \ref{lem:x-infinity}. Now we show that all terms tend to $0$.
 \begin{align*}
   \lim_{\epsilon\to 0} \lim_{T \to \infty}\|J^{(u),T} -J^{(u)}\| &= \lim_{\epsilon\to 0} \lim_{T \to \infty}\left\|\E_{x,\upsilon}\left(\int_0^{\tau \wedge T} e^{-\delta t} u_t \, dt \right)-\E_{x,\upsilon}\left(\int_0^{\tau} e^{-\delta t} u_t \, dt \right)\right\|\\
   &\le \frac{K}{\delta}\lim_{\epsilon\to 0} \lim_{T \to \infty}\left\|\E_{x,\upsilon}\left( e^{-\delta \tau} -e^{-\delta(\tau \wedge T)} \right)\right\|\,,
 \end{align*}
 and using that the last term is bounded we get
 \begin{align*}
   \lim_{\epsilon\to 0} \lim_{T \to \infty}\|J^{(u),T} -J^{(u)}\|
   \le \frac{K}{\delta}\lim_{\epsilon\to 0} \left\|\E_{x,\upsilon}\left(\lim_{T \to \infty}\left( e^{-\delta \tau} -e^{-\delta(\tau \wedge T)} \right)\right)\right\|=0\,.
 \end{align*}
 Analogously we obtain $\lim_{\epsilon\to 0} \lim_{T \to \infty}\|J^{(u),\epsilon}-J^{(u),\epsilon,T}\|=0$.
 For the last term we get
 \begin{align*}
  \lim_{\epsilon\to 0} \lim_{T \to \infty}\|J^{(u),\epsilon,T}-J^{(u),T}\|&=\lim_{\epsilon\to 0} \lim_{T \to \infty}\left\|\E_{x,\upsilon}\left(\int_0^{\tau^\epsilon \wedge T} e^{-\delta t} u_t \, dt \right)-\E_{x,\upsilon}\left(\int_0^{\tau \wedge T} e^{-\delta t} u_t \, dt \right)\right\|\\
  &=\lim_{\epsilon\to 0} \lim_{T \to \infty}\left\|\E_{x,\upsilon}\left(\int_{\tau \wedge \tau^\epsilon \wedge T}^{(\tau \vee \tau^\epsilon) \wedge T} e^{-\delta t} u_t \, dt \right)\right\|\\
  &\le \frac{K}{\delta}\lim_{\epsilon\to 0} \lim_{T \to \infty}\left\|\E_{x,\upsilon}\left(e^{-\delta(\tau \wedge \tau^\epsilon \wedge T)}-e^{-\delta ((\tau \vee \tau^\epsilon) \wedge T)}\right)\right\|\\
  &= \frac{K}{\delta}\lim_{T \to \infty} \left\|\E_{x,\upsilon}\left(\lim_{\epsilon\to 0}\left(e^{-\delta(\tau \wedge \tau^\epsilon \wedge T)}-e^{-\delta ((\tau \vee \tau^\epsilon) \wedge T)}\right)\right)\right\|\,,
 \end{align*}
 where we used boundedness of the term in the last but one row. Now it remains to show that $\tau^\epsilon \to \tau$.
 Proving that $\E\left(\sup_{0\le t \le T}\| \nu^{\epsilon}_t-\nu_t\|^2\right) \to_{\epsilon\to 0} 0$ and $\E\left(\sup_{0\le t \le T}\| X^{(u),\epsilon}_t-X^{(u)}_t\|^2\right) \to_{\epsilon\to 0} 0$ runs along the same lines as in \cite[Proof of Lemma 7.2]{frey2014}.
 Since $X^{(u),\epsilon} \to X^{(u)}$ u.c.p.~we have that along a subsequence $\epsilon(k)$, $X^{(u),\epsilon(k)} \to X^{(u)}$ u.c.a.s.~and hence also $\tau^{\epsilon(k)} \to \tau$ a.s.
 Thus $\lim_{\epsilon\to 0} \lim_{T \to \infty}\|J^{(u),\epsilon,T}-J^{(u),T}\|=0$, which closes the proof.
\end{proof}

Now we are ready to solve \eqref{eq:HJBnuappr} numerically.
We first of all have to restrict our computation domain and therefore choose a sufficiently large number $H$ to approximate the 
domain of the value function by $[0,H]\times [\mu_2,\mu_1]$.\\

To compute the approximate outcomes of our problem, we use policy iteration. Initially, we use a dividend policy of threshold type, since such strategies solve the problem with complete information and thus are good candidates also in our situation. As the initial threshold level we use the convex combination of the threshold levels which are outcomes to the problem with complete information from \citet{sotomayor2011}, as we expect it to be close to the correct solution to our problem:
$b_0(\upsilon):=\frac{\mu_1-\upsilon}{\mu_1-\mu_2}
\bar{b}_2 
+\frac{\upsilon-\mu_2}{\mu_1-\mu_2}\bar{b}_1$,
where $\bar b_1, \bar b_2$ denote the threshold levels in the case with complete information for states $1$ and $2$, respectively.
The initial strategy is given by
$u^{(0)}(x,\upsilon)=K1_{\{x\ge b_0(\upsilon)\}} (x,\upsilon)$.

Note that the mesh we use for our computation is generated such that there are more grid points available where they are needed the most -- between $b_0(\mu_2)$ and $b_0(\mu_1)$.
For more details about the mesh generation we refer to \cite{sz12}.

Now we iteratively apply the following procedure:
\begin{itemize}
 \item For a given strategy $u^{(k)}$ calculate $V^{(k)}$ by solving
\begin{align}\label{eq:numHJB}
(\tilde \cL^G-\delta)V+\bar \epsilon \cD^G_{\upsilon \upsilon}V+u^{(k)}(1-\cD^G_xV)=0\,,
\end{align}
where $\tilde \cL^G$ is the operator $\tilde \cL$ with differentiation 
operators replaced by finite differences, $\cD^G_x$ is 
a finite difference approximation to differentiation with respect to $x$, and
$\cD^G_{\upsilon \upsilon}$ is the finite difference operator replacing the second derivative w.r.t. $\upsilon$.

\item The next iterate $u^{(k+1)}$ is chosen to maximize $u(1-\cD^G_xV^{(k)})$. Thus
$u^{(k+1)}(x,\upsilon)=K1_{\{\cD^G_x V^{(k)}(x,\upsilon)\le 1 \}}$.
\end{itemize}
In our experiments the iteration stops after $6$ steps, since $u^{(6)} \approx u^{(5)}$.\\
The idea behind the construction of the finite difference method is based on the fact that in the discretized setting the diffusion is approximated by Markov chains which locally preserve properties of the original process (cf. \cite[p. 67]{kushner2001}).
The additional term $\bar \epsilon$ is required to guarantee positivity of the scheme and hence to obtain its Markov chain interpretation.
Corresponding convergence results can be found in \cite{kushner1990, kushner2001} and \cite[Chapter IX]{fleming2006}. In \cite[p.~324]{fleming2006} it is noted that the policy iteration converges, yielding value function and associated policy.


\subsection*{Numerical results}
\label{subsec:Results}

We computed both value function and dividend policy for the parameter sets
$\sigma = 1$, $\mu_1=2$, $\mu_2=1$, $\delta =0.5$, $-q_{11}=0.25$, $q_{21}=0.5$, $B=10$, and $K \in \{ 0.2, 0.3, 0.67, 1.8 \}$.
The resulting strategies turn out to be threshold strategies with threshold levels
depending on the estimate of $\mu$. Figure \ref{fig:thresholdplot} shows the resulting threshold levels and compares them to the threshold levels which are outcomes to the dividend maximization problem with constant and unobservable drift, i.e., $Q \equiv 0$, which is studied in \cite{sz12}. Interestingly, while in the case studied herein the threshold level is falling in $\upsilon$ for $K=0.67$, it increases in the Bayesian case. But note that for this parameter choice both curves are rather flat.

The intuition behind the fact that the threshold level grows for some parameter sets, whereas it falls for others is the following.
Usually, the level falls in $\upsilon$ since in the better state the company can pay out dividends earlier as it will recover from it due to the higher drift.
However, for low values of $\upsilon$ the volatility comes more into effect and hence if $K$ is small anyway, then it becomes better to pay dividends even for low values of $x$, because volatility might lead into early ruin. Then, as $\upsilon$ grows it becomes more effective than the volatility and hence soon ruin is not expected anymore and the strategy is designed for a longer living company for higher values of $\upsilon$.

It is interesting that for all of the four parameter sets in the state with the lower drift dividends are paid out more cautiously in the Markov switching case than in the Bayesian case, and in the state with the higher drift it is the other way round. An explanation for this is that since in the Markov switching case there is a chance that the economy gets better, it is the better strategy to wait, if $\upsilon$ is small and pay out dividends at higher values of $\upsilon$.
Hence the state with the lower drift is the state of saving, whereas the other state is the state of spending.
In the Bayesian case the drift does not change and therefore the situation is more balanced.\\

\begin{figure}[ht]
\begin{center}
 \includegraphics[scale=0.6]{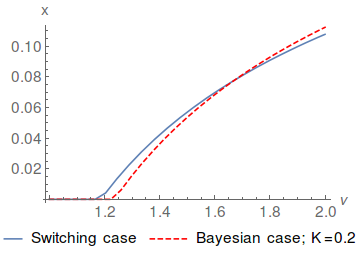}
 \includegraphics[scale=0.6]{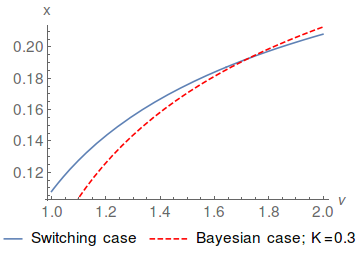}
  
 \includegraphics[scale=0.6]{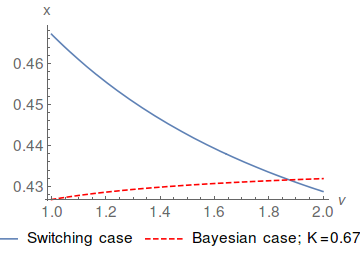}
 \includegraphics[scale=0.6]{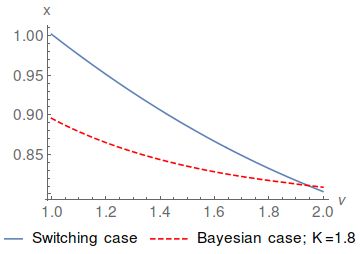}
\end{center}
\caption{The resulting threshold levels for different parameter sets (blue) compared to the threshold levels from the Bayesian case (dashed red).}\label{fig:thresholdplot}
\end{figure}

Figure \ref{fig:valueplot} shows the value function corresponding to $K=1.8$ (but they all look rather similar).

\begin{figure}[ht]
\begin{center}
 \includegraphics[scale=0.6]{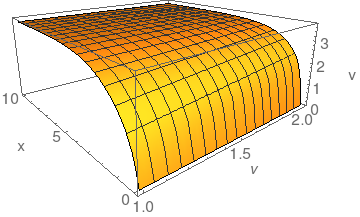}
\end{center}
\caption{The resulting value function for $K=1.8$.}\label{fig:valueplot}
\end{figure}

Figure \ref{fig:valueplot} suggests that the optimal value function is smooth, however proving smoothness is -- due to the degeneracy of the HJB equation, which is highly non-standard -- out of the scope of this paper.\\

Figure \ref{fig:valuecomp} compares the resulting value function to the one from the Bayesian case. We observe that for a smaller estimate of the drift the value is higher in the case with switching,
and for a high estimate of the drift the value is higher in the Bayesian case. This is due to the fact that if a high drift is expected in the Bayesian case, it is more probable that the drift is in fact high,
whereas in the case studied in this paper the drift might change to the worse. For a small estimate of the drift it is exactly the other way round.\\

\begin{figure}[ht]
\begin{center}
\includegraphics[scale=0.6]{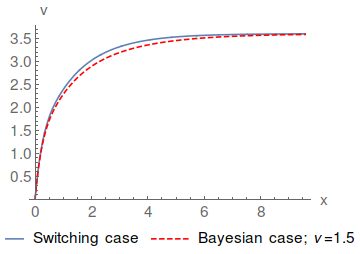}
\includegraphics[scale=0.6]{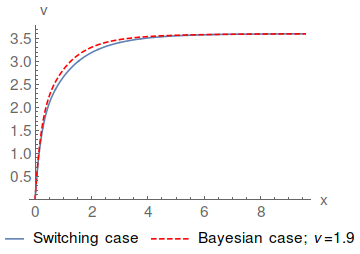}
\end{center}
\caption{The resulting value function for $K=1.8$ (blue) compared to the resulting value function from the Bayesian case (dashed red).}\label{fig:valuecomp}
\end{figure}

One more interesting case to study is that of a high dividend bound $K$. Figure \ref{fig:Kgrows} suggests that for growing parameter $K$ the dividend policy converges to what we expect to be the optimal barrier level in the case of unbounded dividend payments. In future research, it would be of interest to study the singular control problem with unbounded dividend rates.

\begin{figure}[ht]
\begin{center}
\includegraphics[scale=.8]{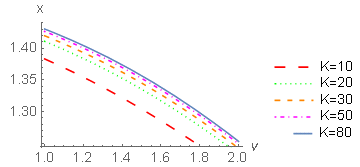}
\end{center}
\caption{The resulting threshold levels for high values of $K$.}\label{fig:Kgrows}
\end{figure}


\subsection*{Admissibility of threshold strategies}
\label{subsec:threshold}

For the Markov switching model under full information \citet{sotomayor2011} find threshold strategies to be optimal. For each state $i$ of the underlying Markov chain they get a constant threshold level $b_i$ such that no dividends are paid below this level and dividends are paid at the maximum rate $K$, if the surplus process exceeds the level. Remember however that in their setup the current state is observable, and therefore it can always be decided which threshold level has to be chosen.\\
In our case, as the current state is estimated in a continuous way, we have only one threshold level. We see that the numerical approximation of the optimal dividend policy is of threshold type with a threshold level $b$ depending on the estimate of the drift, and hence is of the form
\begin{align*}
 u_t=K \, 1_{\{X_t \ge b(\nu_t)\}}(X_t)\,.
\end{align*}


Therefore, it is of considerable importance to study the admissibility of this type of strategies. The question is whether the system
\begin{align}
\label{eq:dynX2ad}X_t&= x + \int_0^t (\nu_s-K \, 1_{\{X_s \ge b(\nu_s)\}}(X_s)) \,ds +\sigma W_t\,,\\
\label{eq:dynpiad}\pi_i(t)&=p_i+\int_0^t \left(q_{Mi} + \sum_{j=1}^{M-1} (q_{ji}-q_{Mi}) \pi_j (s)\right) \, ds + \int_0^t \pi_i (s) \frac{\mu_i-\nu_s}{\sigma} \, dW_s\,, \qquad  i = 1, \dots, M-1 \,,
\end{align}
with $\nu$ as in \eqref{eq:nu}, has a solution. As the drift coefficient of this SDE is discontinuous, classical results from the SDE literature as \cite[Theorem 2.2]{mao2007} 
cannot be applied. However, for threshold levels $b$ which satisfy the Assumptions of \cite[Theorem 3.20]{sz2015b}, we get existence and uniqueness of a unique global strong solution to system \eqref{eq:dynX2ad}, \eqref{eq:dynpiad}.


\section{Summary and conclusion}
\label{sec:Concl}

We have presented a diffusion model for the surplus process of an insurance company, where the drift coefficient changes in response to a change of the economic environment.\\
The change of the economic environment has been modeled by a Markov chain, and uncertainty has been introduced by not allowing to observe the current state of the Markov chain.
We have shown how to overcome uncertainty in this situation by applying a result from stochastic filtering theory. Then we have stated the dividend maximization problem and we have derived the associated HJB equation.\\
We have been able to characterize the solution to the stochastic optimization problem as the unique viscosity solution to this HJB equation.\\
Finally, we have presented an extensive numerical study for the solution to the optimization problem, which suggests that the optimal dividend policy is of threshold type. We have shown that such strategies are indeed admissible using a non-standard result on stochastic differential equations.\\
The main contribution of the current paper is the improvement of the regime switching models that have been studied in the literature by not assuming full information any more. Furthermore, emphasis was put on the numerical study to impart a deeper understanding of the behaviour of the optimal dividend policy for different parameter sets on the one hand, and in comparison to the Bayesian case on the other hand.


\section*{Acknowledgements}

The author thanks Gunther Leobacher (Johannes Kepler University Linz), Stefan Thonhauser (Graz University of Technology) and Ralf Wunderlich (BTU Cottbus-Senftenberg)
for fruitful discussions and helpful advice that improved this paper. \\
Furthermore, the author thanks two anonymous referees for their suggestions.\\
M. Sz\"olgyenyi is supported by the Vienna Science and Technology Fund (WWTF): Project MA14-031.\\
The main part of this paper was written while M. Sz\"olgyenyi was member of the Department of Financial Mathematics and Applied Number Theory, Johannes Kepler University Linz, 4040 Linz, Austria.\\
During this time, M. Sz\"olgyenyi was supported by the Austrian Science Fund (FWF): Project F5508-N26, which is part of the Special Research Program "Quasi-Monte Carlo Methods: Theory and Applications".



\end{document}